\documentclass[a4paper]{llncs}
\usepackage[utf8]{inputenc}
\usepackage{amssymb}
\usepackage{amsmath,bm}
\usepackage{mathrsfs}
\usepackage{xcolor,xspace}
\usepackage{setspace}
\usepackage{scalerel}
\usepackage{theorem}
\usepackage{cite}
\usepackage{url}
\usepackage{array,multirow}
\usepackage{pgfplots}
\usetikzlibrary{plotmarks,patterns}
\usepackage[ruled,vlined,titlenumbered,linesnumbered]{algorithm2e}
\usepackage[colorlinks=true,citecolor=blue,linkcolor=blue,urlcolor=blue]{hyperref}
\usepackage[nolist]{acronym}
\usepackage{orcidlink}

\pgfplotsset{
compat=1.17,
mystyle/.style={
    scale only axis,
    width=0.55\columnwidth,
    height=0.4\columnwidth,
    label style={inner sep=0, font=\normalsize}, 
    tick label style={font=\scriptsize},
    legend style={font=\scriptsize},
    mark size=3,
    major grid style={dashed},
    line width=0.8pt,
    axis line style = thin}
}
\pgfplotsset{
  compat=1.17,
    every axis/.append style={
        scale only axis,
  width=0.55\columnwidth,
  height=0.4\columnwidth,
  label style={inner sep=0, font=\normalsize}, 
  tick label style={font=\scriptsize},
  legend style={font=\scriptsize},
  mark size=3,
  major grid style={dashed},
  line width=0.8pt,
  axis line style = thin}
}

\tikzstyle{SB}    = [color=black, solid]
\tikzstyle{LRS} = [color=red, dash pattern=on 2pt off 4pt on 6pt off 4pt, mark=x, mark options={solid}]
\tikzstyle{FLRS}   = [color=blue, dashed, mark=diamond, mark options={solid}]

\newcolumntype{M}[1]{>{\centering\arraybackslash}m{#1}}

\newcommand{\Fqm}{\ensuremath{\mathbb F_{q^m}}}

\newcommand{\Fqmh}{\ensuremath{\mathbb F_{q^{mh}}}}

\newcommand{\Fq}{\ensuremath{\mathbb F_{q}}}

\newcommand{\F}{\ensuremath{\mathbb F}}

\newcommand{\NN}{\ensuremath{\mathbb{N}}}

\newcommand{\set}[1]{\ensuremath{\mathcal{#1}}}

\newcommand{\Polyring}{\ensuremath{\Fqm[x]}}

\newcommand{\aut}{\ensuremath{\sigma}}

\newcommand{\SkewPolyringZeroDer}{\ensuremath{\Fqm[x,\aut]}}
\newcommand{\MultSkewPolyringZeroDer}{\ensuremath{\Fqm[x,y_1,\dots,y_\intOrder,\aut]}}

\newcommand{\opev}[3]{\ensuremath{{#1}(#2)_{#3}}}

\newcommand{\op}[2]{\ensuremath{\mathcal{D}_{#1}(#2)}}
\newcommand{\opexp}[3]{\ensuremath{\mathcal{D}_{#1}^{#3}(#2)}}
\newcommand{\conj}[2]{\ensuremath{{#1}^{#2}}}

\SetKwComment{Comment}{}{}

\newcommand{\OCompl}[1]{\ensuremath{\mathcal{O}({#1})}}

\newcommand{\defeq}{:=}

\DeclareMathOperator{\wt}{wt}
\DeclareMathOperator{\rk}{rk}

\DeclareMathOperator{\unif}{unif}

\newcommand{\mat}[1]{\ensuremath{\bm{#1}}}

\newcommand{\opVandermonde}[3]{\ensuremath{\mathfrak{m}_{#1}(#2)_{#3}}}
\newcommand{\opMoore}[3]{\ensuremath{\mathfrak{M}_{#1}(#2)_{#3}}}

\renewcommand{\a}{\mathbf a}
\renewcommand{\b}{\mathbf b}
\renewcommand{\c}{\mathbf c}

\newcommand{\f}{\mathbf f}

\newcommand{\p}{\mathbf p}
\newcommand{\q}{\mathbf q}

\newcommand{\x}{\mathbf x}

\newcommand{\B}{\mathbf B}
\newcommand{\C}{\mathbf C}

\newcommand{\E}{\mathbf E}

\renewcommand{\L}{\mathbf L}

\renewcommand{\P}{\mathbf P}

\newcommand{\R}{\mathbf R}
\renewcommand{\S}{\mathbf S}

\newcommand{\0}{\mathbf 0}

\newcommand{\veczeta}{\ensuremath{\boldsymbol{\zeta}}}

\newcommand{\mycode}[1]{\ensuremath{\mathcal{#1}}}

\newcommand{\foldedLinRS}[1]{\ensuremath{\mathrm{F}\mathrm{LRS}[#1]}}

\newcommand{\SumRankWeight}{\ensuremath{\wt_{\Sigma R}}}

\newcommand{\SumRankDist}{d_{\ensuremath{\Sigma}R}}

\newcommand{\shot}[2]{\ensuremath{{#1}^{(#2)}}}

\newcommand{\pe}{\ensuremath{\alpha}}

\newcommand{\degConstraint}{\ensuremath{D}}
\newcommand{\intOrder}{\ensuremath{s}}
\newcommand{\foldPar}{\ensuremath{h}}
\newcommand{\intDim}{\ensuremath{s}}
\newcommand{\shots}{\ensuremath{\ell}}

\newcommand{\lenFLRS}{\ensuremath{N}}
\newcommand{\lenFLRSshot}[1]{\ensuremath{\Lambda}}
\newcommand{\len}{\ensuremath{n}}
\newcommand{\lenShot}[1]{\ensuremath{\lambda}}

\begin{acronym}
 \acro{BMD}{bounded minimum distance}
 \acro{SRS}{skew Reed--Solomon}
 \acro{ISRS}{interleaved skew Reed--Solomon}
 \acro{lclm}{least common left multiple}
 \acro{LRS}{linearized Reed--Solomon}
 \acro{LLRS}{lifted linearized Reed--Solomon}
 \acro{ILRS}{interleaved linearized Reed--Solomon}
 \acro{LILRS}{lifted interleaved linearized Reed--Solomon}
 \acro{FLRS}{folded linearized Reed--Solomon}
 \acro{MSRD}{maximum sum-rank distance}
\end{acronym}

\begin{document}

\title{Efficient Decoding of Folded Linearized Reed--Solomon Codes in the Sum-Rank Metric}

\author{Felicitas Hörmann\,\orcidlink{0000-0003-2217-9753} \and Hannes Bartz\,\orcidlink{0000-0001-7767-1513}}
\institute{
Institute of Communications and Navigation \\ German Aerospace Center (DLR), Germany\\
\email{$\{$felicitas.hoermann, hannes.bartz$\}$@dlr.de}}

\maketitle

\begin{abstract}
 Recently, codes in the sum-rank metric attracted attention due to several applications in e.g. multishot network coding, distributed storage and quantum-resistant cryptography.
 The sum-rank analogs of Reed--Solomon and Gabidulin codes are linearized Reed--Solomon codes.
 We show how to construct $\foldPar$-folded linearized Reed--Solomon (FLRS) codes and derive an interpolation-based decoding scheme that is capable of correcting sum-rank errors beyond the unique decoding radius.
 The presented decoder can be used for either list or probabilistic unique decoding and requires at most $\OCompl{\intDim \len^2}$ operations in $\Fqm$, where $\intDim\leq\foldPar$ is an interpolation parameter and $n$ denotes the length of the unfolded code.
 We derive a heuristic upper bound on the failure probability of the probabilistic unique decoder and verify the results via Monte Carlo simulations.
\end{abstract}

\section{Introduction} \label{sec:introduction}

The sum-rank metric was first encountered in the context of space-time coding~\cite[Sec.~III]{lu2005unified} and can be seen as a hybrid between the Hamming and the rank metric.
Codes in the sum-rank metric are of interest for error control in multishot network coding~\cite{nobrega2010multishot}, for the construction of locally repairable codes~\cite{martinez2019universal} and in the context of quantum-resistant cryptography~\cite{puchinger2020generic}.
The family of \ac{LRS} codes was first described by Mart{\'\i}nez-Pe{\~n}as~\cite{martinez2018skew}, independently studied in~\cite{caruso2019residues}, and fulfills the Singleton-like bound in the sum-rank metric with equality.
A Welch--Berlekamp-like decoder that can correct errors of sum-rank weight $t\leq\lfloor\frac{n-k}{2}\rfloor$, where $n$ is the length and $k$ the dimension of the code, was proposed in~\cite{martinez2019reliable}.
In~\cite{bartz2021fast}, a speed-up was achieved by using approximant bases.
Recently, it was shown in~\cite{bartz2021decoding} and~\cite{bartz2022fast} that \emph{interleaved} \ac{LRS} codes allow to correct sum-rank errors beyond the unique decoding radius.

\emph{Our Contribution:}
We introduce a \emph{folded} variant of LRS codes and provide an interpolation-based algorithm allowing to decode errors of sum-rank weight up to $t < \frac{\intDim}{\intDim+1} \left( \frac{\lenFLRS(\foldPar-\intDim+1)-k+1}{\foldPar-\intDim+1} \right)$, where $\foldPar$ is the blockwise folding parameter, $\lenFLRS$ the code length, $k$ the dimension of the code, and $\intDim\leq\foldPar$ a decoding parameter.
Therefore, our approach allows to correct sum-rank errors beyond the unique decoding radius in quadratic complexity.
Even though the worst-case list size is exponential, we show that a unique solution is obtained with high probability which allows to use the scheme as a probabilistic unique decoder.
We derive a heuristic upper bound on the decoding failure probability and verify the findings by Monte Carlo simulations.
It is worth noting that the proposed decoding scheme generalizes known decoders for folded Reed--Solomon and folded Gabidulin codes in the Hamming and the rank metric, respectively.

\section{Preliminaries} \label{sec:preliminaries}

Let $q$ be a prime power and $\Fq$ a finite field of order $q$.
For any $m \in \NN^{\ast}$, let $\Fqm \supseteq \Fq$ denote an extension field with $q^m$ elements.
We call $\pe \in \Fqm$ \emph{primitive} in $\Fqm$ if it generates the multiplicative group $\Fqm^{\ast} \defeq \Fqm \setminus \{0\}$.

In this paper, we mostly consider matrices whose $\lenFLRS$ columns are divided into $\shots \in \NN^{\ast}$ blocks of the same length $\lenFLRSshot{i} \defeq \frac{\lenFLRS}{\shots} \in \NN^{\ast}$.
Fix $\foldPar \in \NN^{\ast}$ and let $\mat{X}=(\mat{X}^{(1)} \mid \dots \mid \mat{X}^{(\shots)}) \in \Fqm^{\foldPar \times \lenFLRS}$ be a matrix with $\mat{X}^{(i)} \in \Fqm^{\foldPar \times \lenFLRSshot{i}}$ for all $i \in \{1, \ldots, \shots\}$.
Then, the \emph{sum-rank weight} of $\mat{X}$ is defined as
$\SumRankWeight(\mat{X}) \defeq \sum_{i=1}^{\shots} \rk_q\left(\mat{X}^{(i)}\right)$,
where $\rk_q\left(\mat{X}^{(i)}\right)$ is the maximum number of $\Fq$-linearly independent columns of $\mat{X}^{(i)}$.
The \emph{sum-rank distance} of two comparable elements is computed as the sum-rank weight of their difference and forms indeed a metric.
We are concerned with \emph{sum-rank codes} $\mycode{C}$ being subsets of an $\Fqm$-vector space equipped with the sum-rank metric.
If $\mycode{C}$ is an $\Fqm$-linear subspace, the code is called \emph{linear} and its \emph{minimum (sum-rank) distance} is $\SumRankDist(\mycode{C}) = \min \{\SumRankWeight(\c) : \c \in \mycode{C}, \c \neq 0\}$.

Let $\aut$ be an $\Fq$-linear automorphism on $\Fqm$, that is $\aut(a) = a^{q^s}$ for all $a \in \Fqm$ and a particular $s \in \{0, \dots, m-1\}$.
Two elements $a, b \in \Fqm$ are called \emph{conjugate} if there is a $c \in \Fqm^{\ast}$ such that $\conj{a}{c} \defeq \aut(c)ac^{-1} = b$.
The set $\set{C}(a) \defeq \left\{ a^c : c \in \Fqm^{\ast} \right\}$ is called \emph{conjugacy class} of $a$ and $\Fqm$ is partitioned into $q^{\gcd(s,m)}$ of these classes.
If $s = 1$ and $\pe \in \Fqm^{\ast}$ is a primitive element, the set $\{1,\pe,\dots,\pe^{q-2}\}$ contains representatives of all $q-1$ distinct nontrivial conjugacy classes.

The \emph{skew polynomial ring} $\SkewPolyringZeroDer$ (with zero derivation) is defined as the set of polynomials $\sum_i f_i x^i$ with finitely many nonzero coefficients $f_i \in \Fqm$.
It forms a non-commutative ring with respect to ordinary polynomial addition and multiplication determined by the rule $x f_i = \aut(f_i) x$ for all $f_i \in \Fqm$.
We define the \emph{degree} of a skew polynomial $f(x) = \sum_i f_i x^i$ as $\deg(f) \defeq \max \{i : f_i \neq 0\}$ and write $\SkewPolyringZeroDer_{<k}\defeq\{f\in\SkewPolyringZeroDer:\deg(f)<k\}$ for $k \geq 0$.
We further introduce the operator $\op{a}{b} \defeq \aut(b)a$ for any $a, b \in \Fqm$ and its powers $\opexp{a}{b}{i} \defeq \aut^i(b)\aut^{i-1}(a)\cdots\aut(a)a$ for $i \in \NN^{\ast}$.
For a vector $\x = (\x^{(1)} \mid \dots \mid \x^{(\shots)})\in\Fqm^K$ with $\shots$ blocks of length $\kappa \defeq K/\shots \in \NN^{\ast}$, a vector $\a = (a_1, \ldots, a_\shots) \in \Fqm^{\shots}$, and a parameter $d \in \NN^{\ast}$ the \emph{generalized Moore matrix} is defined as
\begin{align}\label{eq:def_gen_moore_mat}
  \opMoore{d}{\x}{\a} &\defeq
  \begin{pmatrix}
    \opVandermonde{d}{\x^{(1)}}{a_1} & \opVandermonde{d}{\x^{(2)}}{a_2} & \cdots & \opVandermonde{d}{\x^{(\shots)}}{a_\shots}
  \end{pmatrix} \in \Fqm^{d \times K}
  ,\\
  \text{where} \quad
  \opVandermonde{d}{\x^{(i)}}{a_i} &\defeq
  \begin{pmatrix}
   x^{(i)}_1 & x^{(i)}_2 & \cdots & x^{(i)}_{\kappa}
   \\
   \op{a_i}{x^{(i)}_1} & \op{a_i}{x^{(i)}_2} & \cdots & \op{a_i}{x^{(i)}_{\kappa}}
   \\[-4pt]
   \vdots & \vdots & \ddots & \vdots
   \\
   \opexp{a_i}{x^{(i)}_1}{d-1} & \opexp{a_i}{x^{(i)}_2}{d-1} & \cdots & \opexp{a_i}{x^{(i)}_{\kappa}}{d-1}
  \end{pmatrix}
  \quad \text{for } 1 \leq i \leq \shots
  .
  \notag
\end{align}
If $\a$ contains representatives of pairwise distinct nontrivial conjugacy classes of $\Fqm$ and $\rk_{q}(\x^{(i)})=\kappa$ for all $1 \leq i \leq \shots$, we have by~\cite[Thm.~2]{martinez2018skew} and~\cite[Thm~4.5]{lam1988vandermonde} that $\rk_{q^m}(\opMoore{d}{\x}{\a}) = \min(d, \shots\kappa)$.

The \emph{generalized operator evaluation} of a skew polynomial $f \in \SkewPolyringZeroDer$ at $b \in \Fqm$ with respect to $a \in \Fqm$ is defined as $\opev{f}{b}{a} = \sum_{i} f_i \opexp{a}{b}{i}$.
Let $a_1,\dots,a_\shots$ be representatives of distinct nontrivial conjugacy classes of $\Fqm$ and consider $n_i$ $\Fq$-linearly independent elements $\zeta_1^{(i)},\dots,\zeta_{n_i}^{(i)} \in \Fqm$ for each $i = 1, \dots, \shots$.
Then any nonzero $f\in\SkewPolyringZeroDer$ satisfying $\opev{f}{\zeta_j^{(i)}}{a_i}=0$ for all $1 \leq j \leq n_i$ and all $1 \leq i \leq \shots$ has degree at least $\sum_{i=1}^{\shots} n_i$ (see e.g.~\cite{caruso2019residues}).

\section{Interpolation-Based Decoding of Folded Linearized Reed--Solomon Codes} \label{sec:decodingFLRS}

Motivated by the results for folded Reed--Solomon codes~\cite{Guruswami2008Explicit,Vadhan2011} and folded Gabi\-dulin codes~\cite{Mahdavifar2012Listdecoding,BartzSidorenko_FoldedGabidulin2015_DCC} we define \ac{FLRS} codes as follows.
We start from a linearized Reed--Solomon code of length $\len \in \NN^{\ast}$ with $\shots \in \NN^{\ast}$ same-sized blocks of length $\lenShot{i} \defeq \frac{\len}{\shots} \leq m$ over $\Fqm$,
and transform each block into an $(\foldPar \times \frac{\lenShot{i}}{\foldPar})$-matrix for a \emph{folding parameter} $\foldPar \in \NN^{\ast}$ dividing $\lenShot{i}$.
\begin{definition}[Folded Linearized Reed--Solomon Codes]\label{def:FLRScodes}
    Consider a primitive element $\pe$ of $\Fqm$ and let $\a=(a_1,\dots,a_\shots) \in \Fqm^{\shots}$ contain representatives of pairwise distinct nontrivial conjugacy classes of $\Fqm$.
    An $\foldPar$-folded linearized Reed--Solomon code of length $\lenFLRS \defeq \frac{\len}{\foldPar}$ and dimension $k \leq \len$ is defined as
    \begin{equation}\label{eq:defFLRScode}
        \foldedLinRS{\a,\pe,\shots,\foldPar;\lenFLRS,k} \defeq
        \left\{
            \left( \C^{(1)}(f) \mid \dots \mid \C^{(\shots)}(f) \right) : f \in \SkewPolyringZeroDer_{<k}
        \right\}
    \end{equation}
    \begin{equation}\label{eq:defFLRScodeblock}
        \text{with }
        \C^{(i)}(f)\defeq
        \begin{pmatrix}
            \opev{f}{1}{a_i} & \opev{f}{\pe^{\foldPar}}{a_i} & \cdots & \opev{f}{\pe^{\lenShot{i}-\foldPar}}{a_i}
            \\
            \opev{f}{\pe}{a_i} & \opev{f}{\pe^{\foldPar+1}}{a_i} & \cdots & \opev{f}{\pe^{\lenShot{i}-\foldPar+1}}{a_i}
            \\
            \vdots & \vdots & \ddots & \vdots
            \\
            \opev{f}{\pe^{\foldPar-1}}{a_i} & \opev{f}{\pe^{2\foldPar-1}}{a_i} & \cdots & \opev{f}{\pe^{\lenShot{i}-1}}{a_i}
        \end{pmatrix}
        \in \Fqm^{\foldPar \times \lenFLRSshot{i}}
    \end{equation}
    for all $i \in \{1, \ldots, \shots\}$.
    We denote the length of a folded block by $\lenFLRSshot{i} \defeq \frac{\lenShot{i}}{\foldPar} = \frac{\len}{\foldPar\shots}$.
\end{definition}
Note that this definition can easily be generalized to different block lengths and more general $\Fq$-linearly independent code locators.
\ac{FLRS} codes are naturally embedded in $\Fqmh$ but linearity is only guaranteed over the subfield $\Fqm$.
\begin{lemma}[Minimum Distance]
    The code $\foldedLinRS{\a,\pe,\shots,\foldPar;\lenFLRS,k}$ has minimum distance $\SumRankDist(\foldedLinRS{\a,\pe,\shots,\foldPar;\lenFLRS,k}) = \lenFLRS - \left\lceil \frac{k}{\foldPar} \right\rceil + 1$.
    It is a \ac{MSRD} code if and only if $\foldPar$ divides $k$.
\end{lemma}

\begin{proof}
    For every nonzero codeword $\C \in \foldedLinRS{\a,\pe,\shots,\foldPar;\lenFLRS,k}$ with message polynomial $f \in \SkewPolyringZeroDer_{<k}$, there are $z, z_1, \dots, z_{\shots} \geq 0$ with $z = \sum_{i=1}^{\shots} z_i$ such that $\SumRankWeight(\C) = \lenFLRS - z$ and $\rk_q(\C^{(i)}) = \lenFLRSshot{i} - z_i$ for $i = 1, \dots, \shots$.
    The column-reduced echelon form of $\C^{(i)}$, whose entries can still be expressed as evaluations of $f$ at evaluation parameter $a_i$, has exactly $z_i$ zero columns.
    In the blockwise reduced matrix are hence $z$ zero columns in total.
    Since the sum of the number of $\Fq$-linearly independent roots of $f$ per evaluation parameter is bounded by its degree, we get $zh \leq k-1$ and equivalently $z \leq \left\lfloor\frac{k-1}{h} \right\rfloor = \left\lceil \frac{k}{\foldPar} \right\rceil - 1$.
    It follows
    $\SumRankDist(\foldedLinRS{\a,\pe,\shots,\foldPar;\lenFLRS,k}) \geq \SumRankWeight(\C) \geq N - \left\lceil \frac{k}{h} \right\rceil + 1$.
    On the other hand, the Singleton-like bound~\cite[Prop. 34]{martinez2018skew} yields
    $\SumRankDist(\foldedLinRS{\a,\pe,\shots,\foldPar;\lenFLRS,k})  \leq N - \frac{k}{h} + 1$
    and the claim follows.
    \qed
\end{proof}

As channel model we consider a sum-rank channel with fixed error weight $t \in \NN$ where the input
$\C\in\foldedLinRS{\a,\pe,\shots,\foldPar;\lenFLRS,k}$ is related to the output $\R$ by $\R = \C + \E \in \Fqm^{\foldPar\times\lenFLRS}$.
The error matrix $\E \in \Fqm^{\foldPar\times\lenFLRS}$ is chosen uniformly at random from the set of all matrices in $\Fqm^{\foldPar\times\lenFLRS}$ having sum-rank weight $t$.
In the following we write
\begin{equation}
    \setlength{\abovedisplayskip}{3pt}
    \setlength{\belowdisplayskip}{\abovedisplayskip}
    \R = \left( \shot{\R}{1} \mid \dots \mid \shot{\R}{\shots} \right)
    \quad \text{and} \quad
    \shot{\R}{i} =
    \begin{pmatrix}
        r_1^{(i)} & r_{\foldPar+1}^{(i)} & \cdots & r_{\lenShot{i}-\foldPar+1}^{(i)}
        \\[-3pt]
        \vdots & \vdots & \ddots & \vdots
        \\
        r_{\foldPar}^{(i)} & r_{2\foldPar}^{(i)} & \cdots & r_{\lenShot{i}}^{(i)}
    \end{pmatrix}
    \in \Fqm^{\foldPar\times\lenFLRSshot{i}}
\end{equation}
for $i \in \{1, \ldots, \ell\}$ and proceed to our interpolation-based decoder.

\subsection{Interpolation Step}

We perform $(s+1)$-variate skew polynomial interpolation with respect to a chosen interpolation parameter $\intDim \in \NN^{\ast}$ with $\intDim \leq \foldPar$.
The set $\set{P}$ of interpolation points is defined by means of a blockwise sliding window approach, whose eligible starting positions are collected in the index set $\set{W}$.
Namely, we consider
\begin{gather}\label{eq:intPointDef}
    \begin{aligned}
         \set{W} &\defeq \left\{ (j-1) \foldPar + l : j \in \{1, \ldots, \lenFLRSshot{i}\}, l \in \{1, \ldots, \foldPar - \intDim + 1\} \right\}
         \\
        \text{and} \quad
        \set{P} &\defeq \left\{ \left( \pe^{w-1}, r_{w}^{(i)}, r_{w+1}^{(i)}, \dots, r_{w+\intDim-1}^{(i)} \right):
        w \in \set{W}, i \in \{1, \ldots, \shots\} \right\}.
    \end{aligned}
\end{gather}

We wish to find a multivariate skew interpolation polynomial of the form
\begin{equation}\label{eq:mult_var_skew_poly_FLRS}
    Q(x,y_1,\dots,y_\intDim)=Q_0(x)+Q_1(x)y_1+\dots+Q_\intDim(x)y_\intDim
    ,
\end{equation}
where $Q_r(x)\in\SkewPolyringZeroDer$ for all $r \in \{0, \ldots, \intDim\}$, that satisfies certain interpolation constraints.
The \emph{generalized operator evaluation} of such a polynomial $Q \in \MultSkewPolyringZeroDer$ at a given interpolation point $(w, i)$ is defined as
\begin{equation}
    \mathscr{E}_{Q}(w, i) \defeq \opev{Q_0}{\pe^{w-1}}{a_i} + \opev{Q_1}{r_{w}^{(i)}}{a_i} + \dots + \opev{Q_\intDim}{r_{w+\intDim-1}^{(i)}}{a_i} \label{eq:defFuncGenOpFLRS}
\end{equation}
where $w \in \set{W}$ and $1 \leq i \leq \shots$ as in~\eqref{eq:intPointDef}.

\begin{problem}[Interpolation Problem] \label{prob:intProblemFLRS}
    For a chosen parameter $\degConstraint \in \NN^{\ast}$ find a nonzero $(s+1)$-variate skew polynomial $Q$ of the form~\eqref{eq:mult_var_skew_poly_FLRS}
    satisfying
    \begin{enumerate}
        \item
            $\mathscr{E}_{Q}(w, i) = 0$ for all $w \in \set{W}$ and $i \in \{1, \ldots, \shots\}$ as well as
        \item
            $\deg(Q_0) < D$ and $\deg(Q_r) < D - k + 1$ for all $r \in \{1, \ldots, \intDim\}$.
    \end{enumerate}
\end{problem}

The second condition of the interpolation problem allows us to write
\begin{equation}
    \label{eq:Qcoefficients}
    \textstyle
    Q_0(x) = \sum_{j=0}^{\degConstraint-1} q_{0, j} x^j
    \quad \text{and} \quad
    Q_r(x) = \sum_{j=0}^{\degConstraint-k} q_{r, j} x^j
    \quad \text{for} \quad
    r \in \{1, \ldots, \intDim\}
\end{equation}
with all coefficients from $\Fqm$.
For each block index $1 \leq i \leq \shots$, we collect all $\lenFLRSshot{i}(\foldPar-\intDim+1)$ interpolation points originating from $\R^{(i)}$ as rows in a matrix $\P_i\in\Fqm^{\lenFLRSshot{i}(\foldPar-\intDim+1)\times (\intDim+1)}$ and denote its columns by $\p_{i, 0}, \ldots, \p_{i, \intDim}$.
Define further $\p_r = (\p_{1, r}^{\top} \ | \ \cdots \ | \ \p_{l, r}^{\top})$ for $0 \leq r \leq \intDim$.
Then, Problem~\ref{prob:intProblemFLRS} can be written as
\begin{gather}\label{eq:intSystem}
    \S\q_I^{\top} = \0 \\
    \text{with} \quad
    \S = \left(
    \begin{array}{c|c|c|c}
        \left(\opMoore{D}{\p_0}{\a}\right)^{\top} & \left(\opMoore{D-k+1}{\p_1}{\a}\right)^{\top} & \cdots & \left(\opMoore{D-k+1}{\p_s}{\a}\right)^{\top}
    \end{array}
    \right)
    \notag \\
    \text{and} \quad
    \q_I = \left(q_{0, 0} \cdots q_{0, D-1} \ | \ q_{1, 0} \cdots q_{1, D-k} \ | \ \cdots \ | \ q_{\intDim, 0} \cdots q_{\intDim, D-k}\right)
    \notag
    .
\end{gather}
The interpolation system~\eqref{eq:intSystem} can be solved using skew Kötter interpolation from~\cite{liu2014kotter} (similar as in~\cite[Sec.~V]{bartz2014efficient}) requiring at most $\OCompl{\intDim \len^2}$ operations in $\Fqm$.

\begin{lemma}[Existence]\label{lem:degContstraintForExistence}
    A nonzero solution to Problem~\ref{prob:intProblemFLRS} exists if
    \begin{equation}
        \degConstraint=\left\lceil\frac{\lenFLRS(\foldPar-\intDim+1)+\intDim(k-1)+1}{\intDim+1}\right\rceil.
    \end{equation}
\end{lemma}

\begin{proof}
    A nontrivial solution of~\eqref{eq:intSystem} exists if less equations than unknowns are involved.
    That is, if $\lenFLRS(\foldPar-\intDim+1) < \degConstraint(\intDim+1)-\intDim(k-1)$.
    \qed
\end{proof}

\begin{lemma}[Roots of Polynomial]\label{lem:decConditionFLRS}
    Define the univariate skew polynomial
    \begin{align}
        P(x) &\defeq Q_0(x)+Q_1(x)f(x)+Q_2(x)f(x)\pe+\dots+Q_\intDim(x)f(x)\pe^{\intDim - 1} \\
        &= Q(x, f(x), f(x)\pe, \dots, f(x)\pe^{\intDim-1}) \in \SkewPolyringZeroDer
        \nonumber
    \end{align}
    and write $t_i \defeq \rk_q(\mat{E}^{(i)})$ for $1 \leq i \leq \shots$.
    Then there exist $\Fq$-linearly independent elements $\zeta_1^{(i)}, \ldots, \zeta_{(\lenFLRSshot{i}-t_i)(\foldPar-\intDim+1)}^{(i)} \in \Fqm$ for each $i \in \{1, \ldots, \shots\}$ such that $\opev{P}{\zeta_j^{(i)}}{a_i} = 0$ for all $1 \leq i \leq \shots$ and all $1 \leq j \leq (\lenFLRSshot{i}-t_i)(\foldPar-\intDim+1)$.
\end{lemma}

\begin{proof}
    Since $\rk_q(\mat{E}^{(i)}) = t_i$, there exists a nonsingular matrix $\mat{T}_i \in \Fq^{\lenFLRSshot{i} \times \lenFLRSshot{i}}$ such that $\mat{E}^{(i)}\mat{T}_i$ has only $t_i$ nonzero columns for every $i \in \{1, \ldots, \shots\}$.
    Without loss of generality assume that these columns are the last ones of $\mat{E}^{(i)}\mat{T}_i$ and define $\veczeta^{(i)} = \L \cdot \mat{T}_i$ with $\L \in \Fqm^{\foldPar \times \lenFLRSshot{i}}$ containing the code locators $1, \dots, \pe^{\lenShot{i}-1}$ (cp.~\eqref{eq:defFLRScodeblock}).
    Note that the first $\lenFLRSshot{i}-t_i$ columns of $\R^{(i)} \mat{T}_i = \C^{(i)}\mat{T}_i + \E^{(i)}\mat{T}_i$ are noncorrupted leading to $(\lenFLRSshot{i}-t_i)(\foldPar-\intDim+1)$ noncorrupted interpolation points according to~\eqref{eq:intPointDef}.
    Now, for each $1 \leq i \leq \shots$, the first entries of the $(\lenFLRSshot{i}-t_i)(\foldPar-\intDim+1)$ noncorrupted interpolation points (i.e. the top left submatrix of size $(\lenFLRSshot{i}-t_i) \times (h-s+1)$ of $\zeta^{(i)}$) are by construction both $\Fq$-linearly independent and roots of $P(x)$.
    \qed
\end{proof}

\begin{theorem}[Decoding Radius] \label{thm:decodingRadius}
    Let $Q(x,y_1,\dots,y_\intDim)$ be a nonzero solution of Problem~\ref{prob:intProblemFLRS}.
    If $t=\SumRankWeight(\mat{E})$ satisfies
    \begin{equation}\label{eq:listDecRegionFLRS}
        \setlength{\abovedisplayskip}{5pt}
        \setlength{\belowdisplayskip}{\abovedisplayskip}
        t<\frac{\intDim}{\intDim+1}\left(\frac{\lenFLRS(\foldPar-\intDim+1)-k+1}{\foldPar-\intDim+1}\right)
        ,
    \end{equation}
    then $P \in \SkewPolyringZeroDer$ is the zero polynomial, that is for all $x \in \Fqm$
    \begin{equation}\label{eq:rootFindingEquationFLRS}
        \setlength{\abovedisplayskip}{5pt}
        \setlength{\belowdisplayskip}{\abovedisplayskip}
        P(x)=Q_0(x)+Q_1(x)f(x)+\!\cdots\!+Q_\intDim(x)f(x)\pe^{\intDim-1}=0.
    \end{equation}
\end{theorem}

\begin{proof}
    By Lemma~\ref{lem:decConditionFLRS}, there exist elements $\zeta_1^{(i)}, \ldots, \zeta_{(\lenFLRSshot{i}-t_i)(\foldPar-\intDim+1)}^{(i)}$ in $\Fqm$ that are $\Fq$-linearly independent for each $i \in \{1, \ldots, \shots\}$ such that $\opev{P}{\zeta_j^{(i)}}{a_i} = 0$ for $1 \leq i \leq \shots$ and $1 \leq j \leq (\lenFLRSshot{i}-t_i)(\foldPar-\intDim+1)$.
    By choosing $\degConstraint \leq (\lenFLRS-t)(\foldPar-\intDim+1)$,
    $P(x)$ exceeds the degree bound from~\cite[Prop.~1.3.7]{caruso2019residues} which is possible only if $P(x)=0$.
    Combining the above inequality with $\lenFLRS(\foldPar-\intDim+1) < \degConstraint(\intDim+1)-\intDim(k-1)$ from the proof of Lemma~\ref{lem:degContstraintForExistence} yields the stated decoding radius.
    \qed
\end{proof}

\subsection{Root-Finding Step}

By Theorem~\ref{thm:decodingRadius}, the message polynomial $f\in\SkewPolyringZeroDer_{<k}$ satisfies~\eqref{eq:rootFindingEquationFLRS} if $t$ satisfies~\eqref{eq:listDecRegionFLRS}.
Therefore, we consider the following root-finding problem.
\begin{problem}[Root-Finding Problem] \label{prob:rootFinding}
    Let $Q \in \MultSkewPolyringZeroDer$ be a nonzero solution of Problem~\ref{prob:intProblemFLRS} and let $t$ satisfy constraint~\eqref{eq:listDecRegionFLRS}.
    Find all skew polynomials $f \in \SkewPolyringZeroDer_{<k}$ that satisfy~\eqref{eq:rootFindingEquationFLRS}.
\end{problem}
Problem~\ref{prob:rootFinding} is equivalent to an $\Fqm$-linear system of equations in the unknown
\begin{equation}
    \setlength{\abovedisplayskip}{4pt}
    \setlength{\belowdisplayskip}{\abovedisplayskip}
    \f \defeq (f_0, \aut^{-1}(f_1), \ldots, \aut^{-k+1}(f_{k-1}))^{\top}
    .
\end{equation}
As e.g. in~\cite{wachter2013decoding,BartzSidorenko_FoldedGabidulin2015_DCC}, we use a basis of the interpolation problem's solution space instead of choosing only one solution $Q$ of system~\eqref{eq:intSystem}.
This improvement is justified by the following result.
\begin{lemma}[Number of Interpolation Solutions]\label{lem:intSolutionDim}
    For $d_I \defeq \dim_{q^m} (\ker(\S))$ with $\S$ defined in~\eqref{eq:intSystem}, it holds $d_I \geq \intDim(\degConstraint-k+1)-t(\foldPar-\intDim+1)$.
\end{lemma}

\begin{proof}
    The first $\degConstraint$ columns of $\S$ are given as $\left(\opMoore{\degConstraint}{\p_0}{\a}\right)^{\top}$.
    Since the $\shots$ blocks of $\p_0$ consist of pairwise distinct powers of $\pe$, the elements of a single block are $\Fq$-linearly independent.
    Hence $\rk_{q^m}(\opMoore{\degConstraint}{\p_0}{\a}) = \min (\degConstraint, \lenFLRS (\foldPar - \intDim + 1)) = \degConstraint$.
    With the absence of an error, the remaining columns consist of linear combinations of the first $\degConstraint$ ones and do not increase the rank.
    If the error $\E$ with $\SumRankWeight(\E) = t$ is introduced, at most $t(\foldPar - \intDim + 1)$ interpolation points are corrupted according to Lemma~\ref{lem:decConditionFLRS}.
    As a consequence, these columns can increase the rank of $\S$ by at most $t(\foldPar - \intDim + 1)$.
    Thus, $\rk_{q^m}(\S) \leq \degConstraint + t(\foldPar - \intDim + 1)$ and the rank-nullity theorem directly yields
    $d_I
    = \degConstraint (\intDim + 1) - \intDim (k-1) - \rk_{q^m}(\S)
    \geq \intDim (\degConstraint - k + 1) - t(\foldPar - \intDim + 1)$.
    \qed
\end{proof}

Let now $Q^{(1)}, \ldots, Q^{(d_I)} \in \MultSkewPolyringZeroDer$ form a basis of the solution space of Problem~\ref{prob:intProblemFLRS} and denote the coefficients of $Q^{(u)}$ by $q_{i,j}^{(u)}$ for all $1 \leq u \leq d_I$ (cp.~\eqref{eq:Qcoefficients}).
Define further the ordinary polynomials
\begin{equation}\label{eq:rf_polys}
    B_j^{(u)}(x) = q_{1, j}^{(u)} + q_{2, j}^{(u)} x + \cdots + q_{\intDim, j}^{(u)} x^{\intDim-1} \in \Polyring
\end{equation}
for $j \in \{0, \ldots, \degConstraint-k\}$ and $u \in \{1, \ldots, d_I\}$ and the additional notations
\vspace*{-5pt}
\begin{align*}
    \b_{j,a} &= \left( \aut^{-a}\left(B_j^{(1)}(\aut^{a}(\pe))\right), \ldots, \aut^{-a}\left(B_j^{(d_I)}(\aut^{a}(\pe))\right) \right)^{\top} \\
    \text{and} \qquad
    \q_a &= \left( \aut^{-a}\left(q_{0,a}^{(1)}\right), \ldots, \aut^{-a}\left(q_{0,a}^{(d_I)}\right) \right)^{\top}
\end{align*}
for $0 \leq j \leq \degConstraint-k$ and $0 \leq a \leq \degConstraint-1$.
Then the root-finding system is given as
\begin{gather}\label{eq:improvedRootFindingSystemFLRS}
    \B \cdot \f = -\q
    \\
    \text{with} \quad
    \B \defeq
    \begin{pmatrix}
        \b_{0,0} & & & \\
        \b_{1,1} & \b_{0,1} & & \\[-3pt]
        \vdots & \b_{1,2} & \ddots & \\[-3pt]
        \b_{\degConstraint-k,\degConstraint-k} & \vdots & & \b_{0,k-1} \\
        & \b_{\degConstraint-k,\degConstraint-k+1} & & \b_{1,k} \\
        & & \ddots & \vdots \\[-3pt]
        & & & \b_{\degConstraint-k,\degConstraint-1}
    \end{pmatrix}
    \quad \text{and} \quad
    \q \defeq
    \begin{pmatrix}
        \q_0 \\
        \vdots \\
        \q_{\degConstraint-1}
    \end{pmatrix}
    .
    \notag
\end{gather}
The root-finding system~\eqref{eq:improvedRootFindingSystemFLRS} can be solved by back substitution in at most $\OCompl{k^2}$ operations in $\Fqm$ since we can focus on (at most) $k$ nontrivial equations from different blocks of $d_I$ rows.
Note also that the transmitted message polynomial $f(x)$ is always a solution of~\eqref{eq:improvedRootFindingSystemFLRS} as long as $t$ satisfies the decoding radius in~\eqref{eq:listDecRegionFLRS}.

\subsection{Interpolation-Based List and Probabilistic Unique Decoding}

The interpolation-based scheme from above can be used for list decoding or as a probabilistic unique decoder.
The list decoder returns all solutions of~\eqref{eq:improvedRootFindingSystemFLRS}.
\begin{lemma}[Worst-Case List Size] \label{lem:worstCaseListSize}
    The list size is upper bounded by $q^{m(\intDim-1)}$.
\end{lemma}

\begin{proof}
    With $d_{RF} \defeq \dim_{q^m}(\ker(\B))$, the list size equals $q^{m \cdot d_{RF}}$ and $d_{RF} = k - \rk_{q^m}(\B)$ due to the rank-nullity theorem.
    Let $\B_{\triangle}$ denote the lower triangular matrix consisting of the first $d_I k$ rows of $\B$.
    Then, $\rk_{q^m}(\B) \geq \rk_{q^m}(\B_{\triangle})$ and the latter is lower bounded by the number of nonzero vectors on its diagonal.
    These vectors are $\b_{0,0}, \ldots, \b_{0,k-1}$ and we focus on their first components while neglecting application of $\aut$.
    Each of them is given as the evaluation of $B_0^{(1)}$ at another conjugate of $\pe$.
    Since $B_0^{(1)}$ can have at most $\intDim - 1$ roots, it follows that at most $\intDim - 1$ of the vectors on the diagonal can be zero.
    Thus, $\rk_{q^m}(\B) \geq k - \intDim + 1$ and, as a consequence, $d_{RF} \leq \intDim - 1$.
    \qed
\end{proof}

Note that, despite the exponential worst-case list size, an $\Fqm$-basis of the list can be found in polynomial time.
Theorem~\ref{thm:listDecoding} summarizes the results for list decoding of \ac{FLRS} codes and Figure~\ref{fig:radiusFLRSopt} illustrates the achievable decoding region.
In particular, the significant improvement of the normalized decoding radius $\tau\defeq t/\lenFLRS$ of \ac{FLRS} codes upon \ac{LRS} codes is shown.
\begin{theorem}[List Decoding] \label{thm:listDecoding}
    Consider a folded linearized Reed--Solomon code $\foldedLinRS{\a,\pe,\shots,\foldPar;\lenFLRS,k}$ and a codeword $\C$ that is transmitted over a sum-rank channel with fixed error weight
    \begin{equation*}
        t < \frac{\intDim}{\intDim+1} \left( \frac{\lenFLRS(\foldPar-\intDim+1)-k+1}{\foldPar-\intDim+1} \right)
    \end{equation*}
    for an interpolation parameter $1 \leq \intDim \leq h$.
    Then, list decoding with a list size of at most $q^{m (\intDim - 1)}$ can be achieved in at most $\OCompl{\intDim \len^2}$ operations in $\Fqm$.
\end{theorem}

\begin{figure}[ht!]
    \centering
\begin{tikzpicture}

\begin{axis}[%
xmin=0,
xmax=1,
xlabel={Code Rate R},
xmajorgrids,
ymin=0,
ymax=1,
ylabel={Fraction of correctable errors $\tau$},
ymajorgrids,
legend style={legend cell align=left,align=left,draw=white!15!black}
]
\addplot [style=SB]
  table[row sep=crcr]{%
0	1\\
0.05	0.95\\
0.1	0.9\\
0.15	0.85\\
0.2	0.8\\
0.25	0.75\\
0.3	0.7\\
0.35	0.65\\
0.4	0.6\\
0.45	0.55\\
0.5	0.5\\
0.55	0.45\\
0.6	0.4\\
0.65	0.35\\
0.7	0.3\\
0.75	0.25\\
0.8	0.2\\
0.85	0.15\\
0.9	0.1\\
0.95	0.05\\
1	0\\
};
\addlegendentry{Singleton Bound};

\addplot [style=LRS]
  table[row sep=crcr]{%
0 0.5\\
0.05  0.475\\
0.1 0.45\\
0.15  0.425\\
0.2 0.4\\
0.25  0.375\\
0.3 0.35\\
0.35  0.325\\
0.4 0.3\\
0.45  0.275\\
0.5 0.25\\
0.55  0.225\\
0.6 0.2\\
0.65  0.175\\
0.7 0.15\\
0.75  0.125\\
0.8 0.1\\
0.85  0.075\\
0.9 0.05\\
0.95  0.025\\
1 0\\
};
\addlegendentry{\ac{LRS} Code};

\addplot [style=FLRS]
  table[row sep=crcr]{%
0.0 0.9615384615384616 \\
0.05 0.8406593406593407 \\
0.1 0.7676470588235295 \\
0.15 0.7037037037037037 \\
0.2 0.6447368421052632 \\
0.25 0.5892857142857143 \\
0.3 0.5357142857142857 \\
0.35 0.4861111111111111 \\
0.4 0.4365079365079365 \\
0.45 0.39090909090909093 \\
0.5 0.34545454545454546 \\
0.55 0.3016304347826087 \\
0.6 0.2608695652173913 \\
0.65 0.22010869565217392 \\
0.7 0.18055555555555555 \\
0.75 0.14583333333333334 \\
0.8 0.1111111111111111 \\
0.85 0.0763888888888889 \\
0.9 0.05 \\
0.95 0.025 \\
1.0 0.0 \\
};
\addlegendentry{FLRS Code};

\end{axis}
\end{tikzpicture}%
    \caption{Normalized decoding radius $\tau \defeq \frac{t}{\lenFLRS}$ vs. code rate $R \defeq \frac{k}{\lenFLRS}$ for an \ac{FLRS} code with $\foldPar = 25$ and optimal decoding parameter $\intDim \leq \foldPar$ for each code rate.}\label{fig:radiusFLRSopt}
    \vspace*{-10pt}
\end{figure}
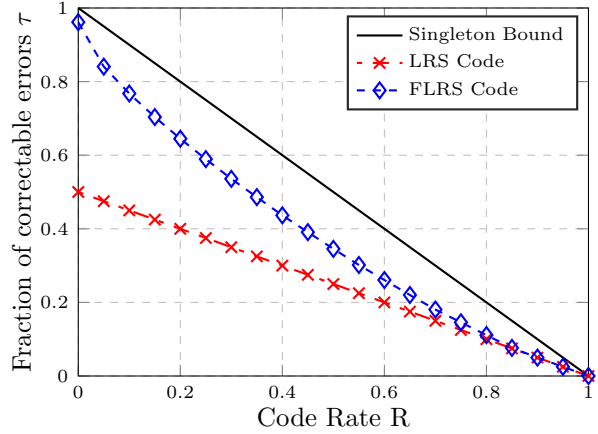

A different concept is probabilistic unique decoding where the decoder either returns a unique solution or declares a failure.
In our setting, a failure occurs exactly when the root-finding matrix $\B$ is rank-deficient.
Similar to~\cite{BartzSidorenko_FoldedGabidulin2015_DCC} we now derive a heuristic upper bound on this probability $\mathbb{P}\left( \rk_{q^m}(\B) < k \right)$.
\begin{lemma}[Decoding Failure Probability] \label{lem:failureProbabilityBound}
    Assume that the coefficients of the polynomials $B_0^{(u)}(x) \in \Polyring$ from~\eqref{eq:rf_polys} for $u \in \{1, \ldots, d_I\}$ are independent and have a uniform distribution among $\Fqm$.
    Then it holds that
    \begin{equation}\label{eq:heuristic_upper_bound}
        \mathbb{P}\left( \rk_{q^m}(\B) < k \right) \lesssim k \cdot \left( \frac{k}{q^m} \right)^{d_I}
        ,
    \end{equation}
    where $\lesssim$ indicates that the bound is a heuristic approximation.
\end{lemma}

\begin{proof}
    Define $\B_{\triangle}$ as in the proof of Lemma~\ref{lem:worstCaseListSize} and note that $\mathbb{P}\left( \rk_{q^m}(\B) < k \right) \leq \mathbb{P}\left( \rk_{q^m}(\B_{\triangle}) < k \right)$ allows to focus on the latter. $\rk_{q^m}(\B_{\triangle}) = k$ is equivalent to all vectors $\b_{0, 0}, \ldots, \b_{0, k-1}$ being nonzero.
    Because application of $\aut$ can be neglected, these vectors can be interpreted as codewords of a Reed--Solomon code.
    The proof of~\cite[Lemma 8]{BartzSidorenko_FoldedGabidulin2015_DCC} deals with this setting and yields the result.
    \qed
\end{proof}

We introduce a threshold parameter $\mu \in \NN^{\ast}$ and enforce $d_I \geq \mu$ which yields a degree constraint $\degConstraint\!=\!\big\lceil \frac{\lenFLRS(\foldPar-\intDim+1)+\intDim(k-1)+\mu}{\intDim+1} \big\rceil$.
Theorem~\ref{thm:probabilisticUniqueDecoding} provides a summary for probabilistic unique decoding of \ac{FLRS} codes incorporating this threshold.

\begin{theorem}[Probabilistic Unique Decoding] \label{thm:probabilisticUniqueDecoding}
    Consider the \ac{FLRS} code $\foldedLinRS{\a,\pe,\shots,\foldPar;\lenFLRS,k}$ and assume that the coefficients of the polynomials $B_0^{(u)}(x)$ for $u \in \{1, \ldots, \mu\}$ are independent and uniformly distributed among $\Fqm$.
    For an interpolation parameter $1 \leq \intDim \leq h$ and a dimension threshold $\mu \in \NN^{\ast}$, transmit a codeword $\C$ over a sum-rank channel with fixed error weight
    \begin{equation}
        \setlength{\abovedisplayskip}{7pt}
        \setlength{\belowdisplayskip}{\abovedisplayskip}
        t \leq \frac{\intDim}{\intDim+1} \left( \frac{\lenFLRS(\foldPar-\intDim+1)-k+1}{\foldPar-\intDim+1} \right) - \frac{\mu}{(\intDim+1)(\foldPar-\intDim+1)}.
    \end{equation}
    Then, $\C$ can be uniquely recovered with complexity $\OCompl{\intDim \len^2}$ in $\Fqm$ and with an approximate probability of at least
    \begin{equation}
        \setlength{\abovedisplayskip}{7pt}
        \setlength{\belowdisplayskip}{\abovedisplayskip}
        1 - k \cdot \left( \frac{k}{q^m} \right)^{\mu}.
    \end{equation}
\end{theorem}

\section{Simulation Results} \label{sec:simulationResults}

We ran simulations in SageMath~\cite{sage} to empirically verify the heuristic upper bound for probabilistic unique decoding from Theorem~\ref{thm:probabilisticUniqueDecoding}.
We chose a $3$-folded FLRS code of length $N = 4$ and dimension $k = 2$ over $\mathbb{F}_{3^6}$ with $\shots = 2$ blocks.
Its minimum distance $4$ implies a unique decoding radius of $1.5$, whereas our probabilistic unique decoder allows to correct errors of weight $t = 2$ for $\intDim = 2$ and $\mu \in \{1, 2\}$ ($t \leq 2.17$ and $t \leq 2$, respectively).
We investigated the case $\mu = 1$ and collected $100$ decoding failures within about $4.23 \cdot 10^7$ randomly chosen error patterns.
The observed failure probability is hence about $2.36 \cdot 10^{-6}$, while the heuristic yields an upper bound of $5.49 \cdot 10^{-3}$.
Note that the parameter set is explicitly designed to obtain an experimentally observable failure probability.

We also tracked the distribution $\chi$ of the coefficients of the polynomials $B_0^{(u)}(x) \in \mathbb{F}_{729}[x]$ from~\eqref{eq:rf_polys} for $1 \leq u \leq \mu$ for multiple transmissions and computed the Kullback--Leibler divergence $D_{KL}$ with respect to the uniform distribution $\unif_{\F_{729}}$, which gives the number of additional bits needed to represent the approximated instead of the actual distribution (see e.g.~\cite[Sec. 2.3]{cover2006elementsOfInformationTheory}).
After $10^6$ transmissions using the above code with $\mu = 1$, the result $D_{KL}(\chi \,||\, \unif_{\F_{729}}) \approx 3.32 \cdot 10^{-4}$ bits shows that the measured distribution $\chi$ is remarkably close to $\unif_{\F_{729}}$.
This justifies the assumption in Theorem~\ref{thm:probabilisticUniqueDecoding}.

\section{Conclusion} \label{sec:conclusion}

We considered the construction of \acl{FLRS} codes and proposed an efficient interpolation-based decoding scheme that is capable of correcting errors beyond the unique decoding radius in the sum-rank metric.
The proposed algorithm can either be used as a (not necessarily polynomial-time) list decoder or as a probabilistic unique decoder that returns a unique solution with high probability.
We analyzed the interpolation-based decoding scheme and derived both an upper bound on the worst-case list size and a heuristic upper bound on the decoding failure probability.
The derivation of an upper bound on the failure probability that incorporates the distribution of the error matrices, a Justesen-like scheme for improved decoding of high-rate codes, and a comparison with decoding schemes for \emph{interleaved} \ac{LRS} codes are subject to future work.

The results in this paper can be extended to obtain more general code constructions over skew polynomial rings with derivations and/or codes with different block sizes.
In particular, \emph{lifted} \ac{FLRS} codes and their properties in the sum-subspace metric can be used for error control in random linear multishot network coding.
The construction of folded \emph{skew} Reed--Solomon codes and the transfer of the presented decoder to the skew metric are other open problems.

\vspace*{-10pt}
\bibliographystyle{splncs04}
\bibliography{references}

\end{document}